\documentclass[10pt,journal,letterpaper]{IEEEtran}

\usepackage{graphicx}
\usepackage{subcaption}
\usepackage{caption}
\usepackage{amsmath}
\usepackage{amssymb}
\usepackage{amsthm} 
\usepackage{accents}
\usepackage{statex}
\usepackage{cases}
\usepackage{mathrsfs}
\usepackage{tikz,pgfplots}
\pgfplotsset{compat=newest}
\usepackage{xcolor}
\usepackage{caption}
\usepackage{setspace}
\usepackage{booktabs}
\usepackage{array}
\usepackage{cite}
\usepackage{enumitem}
\usepackage{algorithm}
\usepackage{algorithmic}
\usepackage{multirow}
\usetikzlibrary{patterns}
\usetikzlibrary{graphs}
\usetikzlibrary{graphs.standard}
\usetikzlibrary{matrix,arrows,fit,backgrounds,mindmap,plotmarks,decorations.pathreplacing}
\newcolumntype{C}[1]{>{\centering\arraybackslash}p{#1}}

\usetikzlibrary{automata}

\newtheorem{theorem}{Theorem}
\newtheorem{definition}{Definition}
\newtheorem{lemma}{Lemma}

\newtheorem{remark}{Remark}
\newtheorem{assumption}{Assumption}

\newlength\figureheight
\newlength\figurewidth

\DeclareMathOperator*{\dist}{dist_+}
\DeclareMathOperator*{\R}{\mathcal{R}}

\makeatletter

\newcommand{\Rmnum}[1]{\expandafter\@slowromancap\romannumeral #1@}
\makeatother

\title{Secure Synchronization of Heterogeneous Pulse-Coupled Oscillators}

\author{Jiaqi Yan and Hideaki Ishii
	\thanks{J. Yan is with the Automatic Control Laboratory, ETH Zurich, Switzerland ({\small jiayan@ethz.ch}). H. Ishii is with the Department of Computer Science, Tokyo Institute of Technology, Japan ({\small ishii@c.titech.ac.jp}).
	}%
	\thanks{This work was also supported by the Swiss National Science Foundation through NCCR Automation under Grant agreement 51NF40\_180545, and by JSPS under
		Grants-in-Aid for Scientific Research Grant No. 22H01508 and 21F40376.}
}
\begin{document}
	\maketitle
	
	\begin{abstract}
In this paper, we consider the synchronization of heterogeneous pulse-coupled oscillators (PCOs), where some of the oscillators might be faulty or malicious. The oscillators interact through identical pulses at discrete instants and evolve continuously with different frequencies otherwise. Despite the presence of misbehaviors, benign oscillators aim to reach synchronization. To achieve this objective, two resilient synchronization protocols are developed in this paper by adapting the real-valued mean-subsequence reduced (MSR) algorithm to pulse-based interactions. The first protocol relies on packet-based communication to transmit absolute frequencies, while the second protocol operates purely with pulses to calculate relative frequencies. In both protocols, each normal oscillator periodically counts the received pulses to detect possible malicious behaviors. By disregarding suspicious pulses from its neighbors, the oscillator updates both its phases and frequencies. The paper establishes sufficient conditions on the initial states and graph structure under which resilient synchronization is achieved in the PCO network. Specifically, the normal oscillators can either detect the presence of malicious nodes or synchronize in both phases and frequencies. Additionally, a comparison between the two algorithms reveals a trade-off between relaxed initial conditions and reduced communication burden.

	\end{abstract}
	
	\section{Introduction}\label{sec:intro}
	Inspired by biological phenomena such as flashing fireflies and contracting cardiac cells,
	pulse-based synchronization becomes a popular topic
	in the systems control community (\hspace{1pt}\cite{mirollo1990synchronization,dorfler2014synchronization,qi2021synthetic}).
	By interacting through simple pulses, the pulse-based synchronization incurs lower energy consumption than its packet-based counterpart. 
	
	In recent years, extensive results have been reported in this area. In \cite{proskurnikov2016synchronization}, the authors have proven 
	that the existence of a directed spanning tree in interaction graph is sufficient and necessary for reaching synchronization. The work \cite{wang2012optimal} maximizes synchronization rate of the PCO network by optimizing the phase response function. 
    Other relevant results can be found in \cite{hu2006scalability,wang2012increasing,qi2021synchronization,kannapan2016synchronization}.
	
	In most of the existing works, the oscillator immediately adjusts its own phase once a pulse signal is received. Moreover, each incoming pulse results in an update on the receiving oscillator's phase. As such, this mechanism increases vulnerabilities of networks in hostile environments. In fact, even a single malicious oscillator is able to destroy the synchronization among normal oscillators by continuously emitting false pulses \cite{wang2020attack}. Motivated by this observation, recent research attention has been paid to developing secure pulse-based algorithms to facilitate the resilient synchronization in an adversarial environment. 
	Most approaches adopt an idea that lets a normal node update its phase only when it receives sufficient number of pulses. Therefore, the communication graph is usually required to be ``sufficiently connected''. For example, \cite{wang2018pulse} and \cite{wang2020global} respectively develop resilient synchronization algorithms in complete and dense graphs, where each oscillator has abundant neighbors. On the other hand, the works \cite{iori2020resilient,iori2021resilient} have characterized the resiliency of algorithms in terms of the network robustness. 
	
	The aforementioned strategies guarantee the resilient synchronization in homogeneous PCO networks where oscillators have identical frequencies. However, they may fail in the network where the oscillators possess frequency heterogeneity--an aspect often expected in real physical systems \cite{nishimura2015frequency,an2009nonidentical}. As such, this paper focuses on the PCO network of heterogeneous oscillators. We aim to propose secure algorithms to enhance the resiliency of this network and facilitate the resilient synchronization of heterogeneous oscillators.

	Our approach is developed based on a class of algorithms which is known as the mean-subsequence reduced (MSR) algorithm. The MSR algorithm is introduced in an early work \cite{dolev1986reaching}. Since then, it has been widely applied to solve the problem of resilient consensus in multi-agent systems, where the state of each agent is represented by a real-valued message  (see \cite{kieckhafer1994reaching,vaidya2012iterative,leblanc2013resilient,dibaji2015consensus,yan2021resilient,yan2020resilient}). The essential idea is that each normal agent discards the most extreme values in its neighborhood and updates its own state with the remaining values only. By doing so, the MSR algorithm overrules the significant effects from the misbehaving nodes and guarantees the normal ones to achieve an agreement even when malicious data is present in the network. 
	
	The MSR algorithm has demonstrated decent performance in tackling real-valued consensus problems. However, the pulse-coupled synchronization problem is characteristic in the following aspects, which prevent the MSR algorithms from being directly applied: 
	
	1) Different from the real-valued consensus settings where nodes always know sources of the received data, the pulses in PCO networks are identical and do not contain the identities of senders. Utilizing this feature, a malicious node has the capacity to transmit numerous deceitful pulses within a short period, thereby deteriorate the synchronization to a greater degree.
	
	2) In the PCO network, the emission of pulses occurs at distinct time instants. As a result, each oscillator should make updates with a set of asynchronous messages. 
	
	3) The pulse-coupled oscillators feature hybrid dynamics: their states are governed by impulsive signals at discrete instants and evolve continuously with their frequencies otherwise. 
	
	4) In heterogeneous PCO networks, there exists a coupling between oscillators' phases and frequencies. This introduces additional complexity to the algorithm design and analysis due to the diversity in frequencies.

 To overcome these difficulties, we tailor the conventional MSR algorithms and propose novel pulse-based synchronization algorithms which are resilient to attacks by extending the approach of \cite{iori2020resilient,iori2021resilient} to the heterogeneous PCO networks. Our main contributions are summarized as follows:	
 
1) Io enhance resiliency of the PCO network, we first develop a detection method to identify malicious nodes emitting multiple pulses per round. Additionally, to counter stealthy attacks that evade detection, we propose two resilient pulse-based protocols using tailored MSR algorithms. Algorithm~\ref{alg:synchronize} relies on packet-based communications for absolute frequencies, while Algorithm~\ref{alg:relative} performs solely through pulses, allowing for the calculation of relative frequencies. Due to the utilization of pure pulses, Algorithm~\ref{alg:relative} consumes less energy, but demands stricter constraints on initial states.
 
2) The proposed algorithms require normal oscillators to update states asynchronously upon receiving a sufficient pulse count, leading to delayed information updates. The hybrid nature and coupling between phases and frequencies present challenges, making conventional MSR algorithm analysis methods unsuitable. To overcome this, we extend findings from \cite{kikuya2017fault} on real-valued messages to pulse-based transmissions.  and leverage hybrid systems theory to tailor them for supporting our proofs.
 
3) Under certain conditions related to initial states and graph structure, our algorithms enable the normal oscillators to either detect misbehaving nodes or achieve resilient synchronization on both phases and frequencies, even in the
presence of malicious signals and frequency heterogeneity. To the best of our knowledge, this study represents the first investigation that considers both heterogeneous dynamics and malicious behaviors in PCO networks.

The rest of this paper is organized as follows. Section~\ref{sec:pre} first presents preliminaries on graph robustness. We then introduce the problem of resilient synchronization in heterogeneous PCO networks in Section~\ref{sec:form}. Resilient synchronization algorithms are developed based on absolute frequencies and relative frequencies respectively in Sections~\ref{sec:alg} and \ref{sec:extend}.
After that, we verify the established results through numerical examples in Section~\ref{sec:simulation} and conclude the paper in Section~\ref{sec:conclude}.

	\section{Preliminaries: Network Robustness}\label{sec:pre}
	Let us consider a digraph $\mathcal{G}=(\mathcal{V},\mathcal{E})$, where $\mathcal{V}$ is the set of nodes, and $\mathcal{E}\subseteq \mathcal{V}\times\mathcal{V}$ is the set of edges. The edge $e_{ij}\in \mathcal{E}$ indicates that node $i$ can directly receive information from node $j$. The set of in-neighbors and out-neighbors of node $i$ are respectively defined as $$\mathcal{N}_i^+\triangleq\{j\in \mathcal{V}|e_{ij}\in \mathcal{E}\}, \;\mathcal{N}_i^-\triangleq\{j\in \mathcal{V}|e_{ji}\in \mathcal{E}\}.$$ Moreover, let us define the in-degree of node $i$ as $d_i\triangleq |\mathcal{N}_i^+|.$
	
	As one might imagine, there is a close coupling between network topology and the maximum number of tolerable faulty nodes. In this paper, we study the resiliency of algorithms in terms of network robustness, as formally defined below:
	\begin{definition}($r$-robustness):
		A digraph $\mathcal{G}=(\mathcal{V},\mathcal{E})$ is $r$-robust, if for any pair of disjoint and nonempty subsets $\mathcal{V}_1, \mathcal{V}_2\subsetneq\mathcal{V}$, at least one of the following statements hold:
		\begin{enumerate}
			\item There exists a node in $\mathcal{V}_1$ such that it has at least $r$ in-neighbors outside $\mathcal{V}_1$.
			\item There exists a node in $\mathcal{V}_2$ such that it has at least $r$ in-neighbors outside $\mathcal{V}_2$.
		\end{enumerate}
	\end{definition}
	
Intuitively, network robustness measures the connectivity of graphs. It requires that for any two disjoint and nonempty subsets of nodes, there is at least one node within these sets that has a sufficient number of in-neighbors from outside. This drives the node away from the values of its own set and towards values of the other set, which is critical for achieving synchronization in the network.

We also need the following lemma:
	
	\begin{lemma}[\hspace{0.5pt}\cite{leblanc2013resilient}]\label{lmm:r'}
		If a digraph $\mathcal{G}=(\mathcal{V},\mathcal{E})$ is $r$-robust, it is also $r'$-robust for any $r'\leq r$.
	\end{lemma}

	\section{Problem Formulation}\label{sec:form}
In this paper, we focus on a network of $N$ pulse-coupled oscillators (PCOs). The PCOs interact over a directed graph $\mathcal{G}=(\mathcal{V},\mathcal{E})$, where any node in $\mathcal{V}=\{1,\cdots,N\}$ corresponds to an oscillator in the network. Each oscillator $i$ is equipped with a phase variable $\phi_i(t)$ that evolves from $0$ to $1$ with absolute frequency $\omega_i(t)$. When the phase reaches $1$, oscillator $i$ emits a pulse (or fires) and resets its phase to $0$. An edge $e_{ij}\in \mathcal{E}$ indicates that oscillator $i$ can receive pulses from oscillator $j$. The received pulses lead to adjustments of the oscillator's phase and frequency, which can be designed to achieve collective behavior, such as synchronization. We assume
\begin{equation}\label{eqn:initialphase}
\min_{i\in\mathcal{V}} \phi_i(0)=0,\;\max_{i\in\mathcal{V}} \phi_i(0)<0.5.
\end{equation}
Without loss of generality, we define the direction in which the phase of the node continuously increases from $0$ to $1$ as clockwise.

Previous research \cite{hu2006scalability,wang2012increasing,wang2012optimal} has proposed various algorithms to synchronize networks of homogeneous oscillators. In contrast, this paper focuses on the synchronization of \textit{heterogeneous} PCOs, where the oscillators have different initial frequencies $\{\omega_i(0)\}$. To avoid loss of generality, we normalize the initial frequencies to be within the range 
\begin{equation}
\Omega \triangleq [1,1+\delta(0)],
\end{equation}
where $\delta(0)$ represents the maximum difference in the initial frequencies of any two oscillators.

%

	\subsection{Attack model}
This paper studies the synchronization of PCOs in an adversarial environment, where some of the oscillators might be behaving incorrectly or even maliciously.
	
To be specific, let $\mathcal{F}$ be the set of misbehaving oscillators. An oscillator $i\in \mathcal{F}$ can either be one that is manipulated by an adversary, or one that fails to follow the pre-defined updating rule. In contrast, normal oscillators always adopt the designed synchronization algorithm. We collect these normal nodes in a set $\mathcal{R}$. For notation convenience, let us define $R$ as the cardinality of $\mathcal{R}$, i.e., the number of normal oscillators in the network.
	
We consider the following attack model, which is widely adopted in the real-valued resilient distributed control problem (see \cite{leblanc2013resilient,sundaram2018distributed,yan2020containment} for examples):
	\begin{assumption}\label{assum:attackmodel}
		The network is under an $f$-local attack\footnote{Note that the $f$-local attack model assumed here is more general than the $f$-total model defined in \cite{leblanc2013resilient}, where the total number of adversarial nodes is upper bounded by $f$.}. That is, for any normal oscillator $i\in\R$, there are no more than $f$ misbehaving nodes in its in-neighborhood, namely, $|\mathcal{F}\cap\mathcal{N}^+_i|\leq f$.
	\end{assumption}
	
This attack model is reasonable since the adversary usually has limited resources and can only manipulate a limited number of nodes. Therefore, the maximum number of misbehaving neighbors should generally be upper bounded. Note that despite this bound, we do not impose any restrictions on the behaviors of misbehaving nodes. They are allowed to emit pulses at arbitrary instants, thus conveying false information about their phases/frequencies to other oscillators. 

\subsection{Resilient synchronization}
	
	The network misbehavior would jeopardize the synchronization among normal oscillators. To address this issue, we aim to propose resilient synchronization algorithms that can work correctly even in the presence of the $f$-local attack. To see this, the concept of a normal containing arc is first introduced as follows:
	
	\begin{definition}[Normal containing arc]\label{def:containarc}
		The normal containing arc is the shortest arc that contains the phases of all normal
		oscillators. Moreover, let $\Delta(t)$ be the length of the normal containing arc at time $t$.
	\end{definition}
 
The main problem addressed in this paper is the resilient synchronization of normal oscillators in the presence of misbehaving nodes:

	\begin{definition}[Resilient synchronization]\label{def:resilientCon}
		The normal oscillators achieve resilient synchronization if the following conditions are satisfied even in the presence of misbehaving nodes:
		\begin{enumerate}
			\item\emph{Safety:} At any time $t$, the absolute frequency of any normal oscillator is in the range $\Omega$.  Moreover, the length of the containing arc is less than $0.5$. That is, for any $i\in\R$, it holds 
			\begin{equation}
				\omega_{i}(t) \in \Omega,\;\Delta(t)<0.5.
			\end{equation}
			\item\emph{Synchronization:} Both phases and frequencies of the normal oscillators asymptotically reach synchronized. That is, for any $i,j\in\R$, it holds
			\begin{equation}
				\begin{split}
					\lim_{t\to\infty} |\phi_{i}(t)-\phi_j(t)|=0,\;\lim_{t\to\infty} |\omega_{i}(t)-\omega_j(t)|=0.
				\end{split}
			\end{equation}
		\end{enumerate}
	\end{definition}

Clearly, normal oscillators reach resilient synchronization if and only if $\Delta(t)$ converges to zero.  
In Definition \ref{def:resilientCon}, the safety condition ensures that the frequency of each normal oscillator stays within the range $\Omega$ at any time, and the length of the containing arc is less than $0.5$, which guarantees that the phases of normal oscillators are not too spread out. The synchronization condition requires that both the phases and frequencies of normal nodes converge to the same value.

	\section{Resilient Synchronization of Pulse-Coupled Oscillators}\label{sec:alg}
	In this section, we will develop a secure algorithm to solve the problem of resilient synchronization among heterogeneous PCOs. We will also prove its efficiency in the presence of $f$-local attacks.
	
\subsection{Design of the resilient algorithm}
Our algorithm is inspired by the well-studied MSR protocol, which solves the real-valued resilient consensus problem (\hspace{1pt}\cite{dolev1986reaching,leblanc2013resilient,dibaji2017resilient}). However, as discussed in Section~\ref{sec:intro}, synchronization in the PCO network is different from the real-valued consensus problem in several aspects, which makes the direct application of the MSR algorithm impossible. Therefore, we introduce subtle modifications to the MSR protocol.

To be specific, each normal oscillator $i$ is assigned two auxiliary variables: a counter $c_i \in\{0,1,\cdots, d_i\}$ to keep track of the number of received pulses and a flag $\gamma_i\in \{0,1\}$ to indicate whether node $i$ has made an update in the current round. Moreover, to achieve frequency synchronization, a set $\mathcal{F}_i$ is introduced to collect the absolute frequencies received from neighbors. For each $i\in\R$, their initial values are set as $\gamma_i=0$, $c_i=0$, and $\mathcal{F}_i=\varnothing$. The whole procedure is presented in Algorithm~\ref{alg:synchronize}.

	\begin{algorithm}[h!] 
		1:\: The phase $\phi_{i}(t)$ evolves from $0$ to $1$.
		
		2:\: Once $\phi_{i}(t)$ reaches $1$, oscillator $i$ sets $\gamma_i=1$, emits a pulse, and resets its phase to $0$. Meanwhile, oscillator $i$ broadcasts its absolute frequency to out-neighbors.
		
		3:\: When oscillator $i$ receives a pulse, it adds the received frequency to $\mathcal{F}_i$ and updates $c_i=c_i+1$. If counter $c_{i}$ becomes $f+1$, it takes $\bar{z}_{i}(t)$ as
		$$
		\bar{z}_{i}(t)= \begin{cases}1-\phi_{i}(t) & \text { if } \phi_{i}(t) \in\left[ 0.5, 1\right) ,\\ 0 & \text { otherwise. }\end{cases}
		$$ Else if the counter $c_{i}$ becomes $d_{i}-f$, it sets $\underline{z}_{i}(t)$ as
		$$
		\underline{z}_{i}(t)= \begin{cases}-\phi_{i}(t) & \text { if } \phi_{i}(t) \in\left[0, 0.5\right), \\ 0 & \text { otherwise. }\end{cases}
		$$
		4:\: When $\phi_{i}(t)=0.5$ and $\gamma_i=1$, oscillator $i$ executes the following:
		
		\begin{enumerate}[leftmargin = 30 pt]
			\item[a):] If $c_{i}>d_{i}$, oscillator $i$ detects an attack and informs the system. It does not make any update and sets $\phi_{i}\left(t\right)=\phi_{i}(t^{-})$ and $\omega_{i}\left(t\right)=\omega_{i}(t^{-})$.
			
			\item[b):] If $c_{i} \leq d_{i}$, let us denote 
			\begin{equation}\label{eqn:f0}
				f_i = f-(d_{i} - c_{i}).
			\end{equation}
			Oscillator $i$ updates the phase $\phi_{i}(t)$ as
			\begin{align}\label{eqn:phaseupdate}
				\phi_{i}\left(t\right)&=\phi_{i}(t^{-})+ \frac{\bar{z}_{i}(t)+\underline{z}_{i}(t)}{2}.
			\end{align}
			Then, it sorts the values in $\mathcal{F}_i$ and removes $f_i$ largest and $f_i$ smallest values in $\mathcal{F}_i$. Denoting by $\mathcal{J}_i(t)$ the set of remaining states,
			oscillator $i$ updates its frequency as
			\begin{equation}\label{eqn:frequencyupdate}	
				\omega_i\left(t\right)= a_{ii}(t)\omega_i(t^{-})+\sum_{\omega_j(t^-)\in \mathcal{J}_i(t)} a_{ij}(t) \omega_j(t^-),			
			\end{equation}
			where each weight is lower bounded by $\alpha>0$ and $a_{ii}(t)+\sum_{\omega_j(t^-)\in \mathcal{J}_i(t)} a_{ij}(t)=1.$
		\end{enumerate}
		
		5:\: Oscillator $i$ resets $\gamma_i=0$, $c_i=0$, and $\mathcal{F}_i=\varnothing$.\\
		\caption{Resilient synchronization of heterogeneous pulse-coupled oscillators by absolute frequencies}
		\label{alg:synchronize}
	\end{algorithm}

In many existing works, an update on phase occurs immediately upon receiving a pulse. Therefore, each incoming pulse triggers an adjustment on the receiving oscillator's phase. However, according to \cite{wang2018pulse}, this mechanism is not resilient to network misbehaviors. In contrast, Algorithm~\ref{alg:synchronize} allows normal oscillators to update asynchronously and use delayed information. This means that the time instants of the node's receiving pulses and making updates are different. Furthermore, with the counter $c_i$, the proposed algorithm improves resilience in the sense that each node adjusts its phase only when a sufficient number of pulses are received. The proposed algorithm also uses an MSR-based protocol \eqref{eqn:frequencyupdate} to renew the node's frequency. We will show that this algorithm can resiliently facilitate synchronization in the network of heterogeneous PCOs even with the coupling between phases and frequencies.

\begin{remark}
In Algorithm~\ref{alg:synchronize}, the oscillators are required to broadcast their absolute frequencies $\{\omega_i(t)\}$ to neighbors. This requires packet-based communications to transmit real-valued messages. Moreover, accessing the absolute frequency for each oscillator can also be challenging in practice, as each oscillator operates with its own clock, making only the relative frequency available. To address these challenges, we will demonstrate in Section~\ref{sec:extend} how resilient synchronization can be achieved with relative frequencies by simply using the pulse-based communication. 
\end{remark}

\subsection{Performance analysis}
This subsection provides a performance analysis of Algorithm~\ref{alg:synchronize}. It is important to note that the pulse in this network does not contain any information on its source or destination, which means that normal oscillators may not be able to identify irregular behavior in their neighbors. If a malicious node keeps continuously emitting pulses, it can effectively prevent normal nodes from reaching synchronization. However, Algorithm~\ref{alg:synchronize} includes a mechanism to easily detect such misbehavior.
To be specific, each normal node $i$ in the network monitors the number of pulses it receives within a given period of time. If a normal node receives a number of pulses greater than $d_i$ within a round, it will conclude that there are malicious nodes present and inform the system. This mechanism allows our algorithm to detect the presence of malicious nodes with high accuracy.

However, the algorithm is not able to detect more sophisticated  stealthy attacks\footnote{The stealthy attack in this paper is defined in a similar  way to that in \cite{iori2020resilient}.}  where malicious nodes emit pulses intermittently or with low frequency. In the rest of this section, we will show that the proposed algorithm can also efficiently work against these stealthy attacks and facilitate resilient synchronization among normal nodes. To demonstrate this, we first make the following assumption:
\begin{assumption}\label{ass:network}
 The network $\mathcal{G}=(\mathcal{V},\mathcal{E})$ is $(2f+1)$-robust.
\end{assumption}

Assumption~\ref{ass:network} has been proven to be sharp in achieving real-valued resilient consensus (\hspace{1pt}\cite{leblanc2013resilient}) and is therefore also applied here to the problem of resilient consensus to the heterogeneous PCO network. It is easy to verify that, under Assumption~\ref{ass:network}, each normal node has at least $2f+1$ in-neighbors, i.e., $d_i\geq 2f+1$.

Although the oscillators evolve in continuous time, their behaviors are completely dominated by the event times of emitting the pulses and making the updates. Therefore, we view the entire PCO network as a discrete-time system. To be specific, let $\mathbb{T}$ be the set of all event times when the normal or malicious nodes fire, or when the normal nodes update. If multiple events happen at the same time, they would be counted independently. Let us arrange the elements in $\mathbb{T}$ in an increasing order and denote the event time sequence by $\left\{t_{k} \right\}$, where $0=t_{0}\leq t_{1}\leq \cdots\leq t_{k}\leq \cdots$. For simplicity, we define 
\begin{equation}
\begin{split}
\phi_i(k)\triangleq \phi_i(t_k), \; \omega_i(k)\triangleq\omega_i(t_k).
\end{split}
\end{equation}
Let us denote the minimum and maximum absolute frequencies of normal oscillators respectively as
	\begin{equation}
		\begin{split}
			&\underline{\omega}(k) \triangleq \min_{i\in\R} \omega_i(k),\; \bar{\omega}(k) \triangleq \max_{i\in\R} \omega_i(k).
		\end{split}
	\end{equation}
In the presence of stealthy attacks, there exists a finite $\bar{k}$ such that each normal node makes at least one update within $\bar{k}$ events.
Note that each normal or stealthy attacking oscillator can contribute at most two to the set $\mathbb{T}$ per round. Hence, we can conclude 
\begin{equation}\label{eqn:kbar}
\bar{k}\leq 2N.
\end{equation}
We further define
\begin{equation}\label{eqn:frequency}
	\begin{split}
		&m_\omega(k) \triangleq \min\{\underline{\omega}(k-\bar{k}+1),\underline{\omega}(k-\bar{k}+2),\cdots,\underline{\omega}(k)\}, \\
		&M_\omega(k) \triangleq \max\{\bar{\omega}(k-\bar{k}+1),\bar{\omega}(k-\bar{k}+2),\cdots,\bar{\omega}(k)\}, 
	\end{split}
\end{equation}
and
\begin{equation}
	\begin{split}
		\delta(k) &\triangleq M_\omega(k)-m_\omega(k).
	\end{split}
\end{equation}
	
With these preparations, we are now ready to demonstrate the convergence of Algorithm~\ref{alg:synchronize}.
Given the network of heterogeneous oscillators, we will show the synchronization of frequencies and phases separately. It is worth noting that the coupling between frequencies and phases further complicates the analysis.
To tackle this issue, we first introduce the following lemma:
	\begin{lemma}\label{lmm:frequency}
	Consider any event instant $k\geq 0$. Suppose that Assumptions~\ref{assum:attackmodel} and \ref{ass:network} hold, and $\Delta(t)<0.5$ for any $t\in\mathbb{Z}_{\geq 0}$ and $t\leq k$. Under stealthy attacks, it follows
\begin{equation}\label{eqn:M}
m_\omega(k) \geq m_\omega(k-1),\;M_\omega(k) \leq M_\omega(k-1),
\end{equation}
and
\begin{equation}\label{eqn:epsilon}
\delta(k)\leq \bigg(1-\frac{\alpha^{\bar{k}R}}{2}\bigg)^{\lfloor k/(\bar{k}R)\rfloor}\delta(0).
\end{equation}
	\end{lemma}
	\begin{proof}
Let us consider the update of each normal node $i\in\R$. To make our discussion more precise, we define the term \textit{round} for node $i$, denoting the period between two consecutive resetting times of counter $c_i$. In other words, a round is the time period from the moment when node $i$ resets its counter $c_i$ and starts counting the pulses until the moment when its counter is reset again.

Since $\Delta(t)<0.5$ for any $t\in [0,k]$, it is easy to verify that node $i$ receives one and only one pulse from each of its normal in-neighbors within every round. Therefore, $c_i\geq d_i-f$. On the other hand, since the attack is stealthy, namely, it is not detected by any normal node running Algorithm~\ref{alg:synchronize}, we conclude that $c_i\leq d_i$. In view of \eqref{eqn:f0}, one has $f_i\in[0,f]$. By virtue of Lemma~\ref{lmm:r'}, the network is $(2f_i+1)$-robust. Moreover, within this round, there exist at most $f_i$ malicious values in $\mathcal{F}_i$. 

According to Algorithm~\ref{alg:synchronize}, node $i$ removes
$2f_i$ values from $\mathcal{F}_i$, where $f_i$ from above and $f_i$ from below. Notice that by doing so, those
faulty values outside the interval $[m_\omega(t),M_\omega(t)]$ are all ignored in the update \eqref{eqn:frequencyupdate}, namely, $\omega_j(t^-)\in[m_\omega(t),M_\omega(t)],\;\forall \omega_j(t^-)\in \mathcal{J}_i(t)$. In other words, the absolute frequency of node $i$ is affected by only the values within $[m_\omega(t),M_\omega(t)]$. By \eqref{eqn:frequencyupdate}, the new frequency $\omega_i\left(t\right)$ is a convex combination of values in this interval. Since the frequency of any normal node keeps unchanged until next event, one concludes that \eqref{eqn:M} holds.

Next, for any $\epsilon\in\mathbb{R}$ and $\bar{t}\geq t$, let us define
\begin{equation}\label{eqn:V^MandV^m}
	\begin{split}
		\mathcal{R}_\omega^m(\bar{t},t,\epsilon)\triangleq\{i\in\mathcal{R}: \omega_{i}(\bar{t})<m_\omega(t)+\epsilon\},\\
		\mathcal{R}_\omega^M(\bar{t},t,\epsilon)\triangleq\{i\in\mathcal{R}: \omega_{i}(\bar{t})>M_\omega(t)-\epsilon\}.
	\end{split}
\end{equation}
That is, $\mathcal{R}_\omega^m(\bar{t},t,\epsilon)$ [resp. $\mathcal{R}_\omega^M(\bar{t},t,\epsilon)$] includes all normal nodes with absolute frequency lower [resp. greater] than $m_\omega(t)+\epsilon$ [resp. $M_\omega(t)-\epsilon$] at time $\bar{t}$. 

Suppose that at time $t$, it holds that $\underline{\omega}(k) \neq \bar{\omega}(k)$. Therefore, $\delta(t)>0$. Let us define $\epsilon(t) = \delta(t)/2$. It is easy to verify that $\mathcal{R}_\omega^M(t,t,\epsilon(t))$ and $\mathcal{R}_\omega^m(t,t,\epsilon(t))$ are disjoint and both nonempty. Furthermore, we recursively define $\epsilon(t+1)=\alpha\epsilon(t)<\epsilon(t)$. 
We will prove that, the cardinality of $\mathcal{R}_\omega^M(\bar{t},t,\epsilon(\bar{t}))$ is a nonincreasing function of time $\bar{t}=t,t+1,\cdots,t+\bar{k}$. To see this, it is sufficient to show that for any normal node $j$, if it is not in the set $\mathcal{R}_\omega^M(\bar{t},t,\epsilon(\bar{t}))$, then it is also not in the set $\mathcal{R}_\omega^M(\bar{t}+1,t,\epsilon(\bar{t}+1))$. This clearly holds if node $j$ does not update its frequency at time $\bar{t}$. On the other hand, if node $j$ makes an update at $\bar{t}$, it follows that
\begin{equation}\label{eqn:upperbound}
	\begin{split}
		\omega_j\left(\bar{t}+1\right)&= a_{jj}(\bar{t})\omega_j(\bar{t})+\sum_{\omega_l(\bar{t})\in \mathcal{J}_j(\bar{t})} a_{jl}(\bar{t}) \omega_l(\bar{t})\\
		&\leq \alpha[M_\omega(t)-\epsilon(\bar{t})]+ (1-\alpha) M_\omega(t) \\&= M_\omega(t)-\alpha\epsilon(\bar{t})\\&=M_\omega(t)-\epsilon(\bar{t}+1),
	\end{split}
\end{equation}
where the inequality holds by \eqref{eqn:M} and we place the largest possible weight on $M_\omega(t)$. Therefore, it holds $j\notin\mathcal{R}_\omega^M(\bar{t},t,\epsilon(\bar{t}))$. Similarly, we can check that the cardinality of $\mathcal{R}_\omega^m(\bar{t},t,\epsilon(\bar{t}))$ is also nonincreasing in $\bar{t}=t,t+1,\cdots,t+\bar{k}$.

Next, we show that at least one of the sets $\mathcal{R}_\omega^M(\bar{t},t,\epsilon(\bar{t}))$ and $\mathcal{R}_\omega^m(\bar{t},t,\epsilon(\bar{t}))$ will become empty after finite events. To see this, recall that the network is $(2f+1)$-robust. Therefore, there exists a normal node, labeled as $i$, in either $\mathcal{R}_\omega^M(t,t,\epsilon(t))$ or $\mathcal{R}_\omega^m(t,t,\epsilon(t))$ such that it has at least $2f+1$ in-neighbors from outside the set it belongs to.

Without loss of generality, let $i\in \mathcal{R}_\omega^M(t,t,\epsilon(t))$. It has no less than $2f+1$ in-neighbors in $\mathcal{V}\backslash\mathcal{R}_\omega^M(t,t,\epsilon(t))$. Under Assumption~\ref{assum:attackmodel}, at least $f+1$ of these in-neighbors must be normal. As discussed above, these nodes will remain outside of $\mathcal{R}_\omega^M(\bar{t},t,\epsilon(\bar{t}))$ for $\bar{t}=t,\cdots,t+\bar{k}$. On the other hand, notice that node $i$ makes at least one update in $\bar{t}=t,\cdots,t+\bar{k}$, when it ignores at most $f_i\leq f$ values outside of $\mathcal{R}_\omega^M(\bar{t},t,\epsilon(\bar{t}))$. Therefore, at least one element in $\mathcal{J}_i(\bar{t})$ is upper bounded by $M_\omega(t)-\epsilon(\bar{t})$. Consequently, it follows that
\begin{equation}
	\begin{split}
		\omega_i\left(\bar{t}+1\right)&= a_{ii}(\bar{t})\omega_i(\bar{t})+\sum_{\omega_j(\bar{t})\in \mathcal{J}_i(\bar{t})} a_{ij}(\bar{t}) \omega_j(\bar{t})\\
		&\leq (1-\alpha) M_\omega(t) + \alpha[M_\omega(t)-\epsilon(\bar{t})]\\&=M_\omega(t)-\epsilon(\bar{t}+1).
	\end{split}
\end{equation}
Therefore, after $\bar{k}$ steps, node $i$ will be outside of $\mathcal{R}_\omega^M(t+\bar{k},t,\epsilon(t+\bar{k}))$. Similarly, if $i\in \mathcal{R}_\omega^m(t,t,\epsilon(t))$, we can obtain an analogous result that node $i$ will be outside of  $\mathcal{R}_\omega^m(t+\bar{k},t,\epsilon(t+\bar{k}))$. We therefore conclude that $|\mathcal{R}_\omega^M(t+\bar{k},t,\epsilon(t+\bar{k}))|+|\mathcal{R}_\omega^m(t+\bar{k},t,\epsilon(t+\bar{k}))|<|\mathcal{R}_\omega^M(t,t,\epsilon(t))|+|\mathcal{R}_\omega^m(t,t,\epsilon(t))|.$

As long as both $\mathcal{R}_\omega^M(t+\bar{k},t,\epsilon(t+\bar{k}))$ and $\mathcal{R}_\omega^m(t+\bar{k},t,\epsilon(t+\bar{k}))$ are nonempty, we can repeat the above analysis and conclude that the cardinality of at least one of these two sets will decrease by $1$ after $\bar{k}$ events. Moreover, since both $|\mathcal{R}_\omega^M(t,t,\epsilon(t))|$ and $|\mathcal{R}_\omega^m(t,t,\epsilon(t))|$ are upper bounded by $R$, one of these sets would be empty after $\bar{k}R$ steps. Namely, one of following statements must be true:
\begin{enumerate}
	\item $\mathcal{R}_\omega^M(t+\bar{k}R,t,\epsilon(t+\bar{k}R))=\varnothing,$
	\item $\mathcal{R}_\omega^m(t+\bar{k}R,t,\epsilon(t+\bar{k}R))=\varnothing.$
\end{enumerate}
Without loss of generality, we assume that the first statement holds. By \eqref{eqn:V^MandV^m}, at time $t+\bar{k}R$, the absolute frequencies of all normal nodes are upper bounded by $M_\omega(t)-\epsilon(t+\bar{k}R)$, namely,
$M_\omega(t+\bar{k}R)\leq M_\omega(t)-\epsilon(t+\bar{k}R).$ On the other hand, from \eqref{eqn:M}, we have
$m_\omega(t+\bar{k}R)\geq m_\omega(t).$
Therefore, one concludes that                                                                                           
\begin{equation}
	\begin{split}
		\delta(t+\bar{k}R)&\leq \delta(t)-\epsilon(t+\bar{k}R)=
		\bigg(1-\frac{\alpha^{\bar{k}R}}{2}\bigg)\delta(t).
	\end{split}
\end{equation}
On the other hand, in view of \eqref{eqn:M}, we obtain for any $t$ that
\begin{equation}
	\delta(t+1)\leq \delta(t).
\end{equation}Then it is not difficult to verify that \eqref{eqn:epsilon} holds.
\end{proof}

A requirement for holding Lemma~\ref{lmm:frequency} is that the length of the containing arc is always within $0.5$, i.e., $\Delta(t)<0.5$. This raises a natural question: under what condition can we enforce this requirement? Specifically, in the heterogeneous PCO network, differences in frequencies and phases can affect each other. Therefore, in the following discussion, we will investigate how this coupling influences $\Delta(t)$.

To simplify the analysis, let us introduce a ``virtual" node $s$. The virtual node evolves freely from $0$ to $1$ without any jumps in its phase. At each instant $k$, we respectively denote the phase and absolute frequency of $s$ by $\phi_s(k)$ and $\omega_s(k)$, where $\phi_s(0)=0$. On the other hand, $\omega_s(k)$ is equal to the smallest absolute frequency among normal nodes at time $k$, namely,
$\omega_s(k) = m_\omega(k)$. It is important to note that this virtual node is used solely for analysis purposes and need not be known by any nodes in the network.

Similar to Definition~\ref{def:resilientCon}, we define a \textit{virtual containing arc} as below:
\begin{definition}[Virtual containing arc]\label{def:virtualarc}
	The virtual containing arc is the shortest arc that contains the phases of all normal
	oscillators and the virtual node.
\end{definition}
We define the two ends of the arc, following the clockwise direction, as the head and the tail, respectively. Then, we have the following result:

\begin{lemma}\label{lmm:alphabeta}
	Consider any event instant $k\geq 0$. Suppose that Assumptions~\ref{assum:attackmodel} and \ref{ass:network} hold, and $\Delta(t)<0.5$ holds for any $t\in\mathbb{Z}_{\geq 0}$ and $t\leq k$. If the normal nodes perform Algorithm~\ref{alg:synchronize} under stealthy attacks, then the virtual node $s$ is the tail of the virtual containing arc at time $k$.
\end{lemma}
\begin{proof}
	In order to prove this lemma, it is sufficient to show that the virtual node $s$ is always behind all normal nodes in the clockwise direction. We will prove this by induction. It is clear that the statement holds at $k=0$. We assume that it also holds for any $t<k$. If no normal nodes update at the $k$-th event, which implies that all normal nodes evolve from $k-1$ to $k$ without adjustment on phases and frequencies, then the statement is easily verified since node $s$ has less or equal absolute frequency than all normal oscillators. On the other hand, let us consider the case where a normal node, labeled as $i$, makes an update at time $k$. Since $\Delta(t)<0.5$ holds for all $t< k$, node $i$ receives one and only one pulse from each of its normal in-neighbors within every round. Given the $f$-local attack model, we conclude that the number of pulses it receives within each round is more than $d_i-f$. Here, we introduce two indices $\alpha,\beta\in\mathcal{R}$ for node $i$. Specifically, suppose that node $i$ receives the $(f+1)$-th pulse at time $k_1$. If it holds that $\bar{z}_{i}(k_1)= 1-\phi_{i}(k_1)$, let node $\alpha$ be the one that emits the $(f+1)$-th pulse. Otherwise, let $\alpha$ be node $i$ itself. Similarly,
	suppose that the $(d_i-f)$-th pulse is transmitted at time $k_2$. Let node $\beta$ be the one that emits this pulse if $\underline{z}_{i}(k_2)= -\phi_{i}(k_2)$, and be node $i$ otherwise. Assumption~\ref{ass:network} implies that $d_i\geq 2f+1$. Therefore, both $\alpha$ and $\beta$ exist and $k_1\leq k_2<k$.
	
	It is worth noticing that under the $f$-total attack model, at time $k_2$, among
	all the normal in-neighbors of node $i$, the one taking the least relative phase with $s$ will be the $(d_i-f)$-th one or later to emit a pulse. In either case, one can verify that $\beta$ is on the normal containing arc at time $k_2$. Consequently, 
	when making the update, the phase of node $i$ decreases by at most 
	\begin{equation}
	\begin{split}
	\frac{\underline{z}_{i}(k)}{2}&=\frac{1}{2}\dist(\phi_{i}(k_2),\phi_{\beta}(k_2))\\&\leq \frac{1}{2}\dist(\phi_{i}(k_2),\phi_{s}(k_2))\\&\leq \frac{1}{2}\dist(\phi_{i}(k^-),\phi_{s}(k^-)),
	\end{split}
	\end{equation} 
	where the first inequality holds as $s$ in the tail of of the virtual containing arc at time $k_2$ by assumption, and the last inequality holds as $\omega_s(t)\leq \omega_i(t)$ for all $t$. Therefore, node $i$ remains ahead of node $s$ even after the phase update. The other normal nodes evolve freely from $k-1$ to $k$ and thus they also continue to be ahead of node $s$. This completes the proof.
		
\end{proof}

For each $i\in\R$, we define its \textit{relative phase} with respect to $s$ as
\begin{equation}
	r_i(t) \triangleq \dist(\phi_i(t),\phi_s(t)),
\end{equation}
where given any two phases $\phi_i, \phi_j \in[0,1]$, $\dist\left(\phi_i, \phi_j\right)$ is given by
$$
\dist\left(\phi_i, \phi_j\right)=\begin{cases}
\phi_i - \phi_j, \text{if } \phi_i\geq \phi_j,\\
1- \phi_j+\phi_i, \text{ otherwise}.
\end{cases}
$$
Combining the above definition with Lemma~\ref{lmm:alphabeta}, it is easy to see that $r_i(t)$ is the distance between $\phi_i(t)$ and $\phi_s(t)$ in the clockwise direction.
Similar to \eqref{eqn:frequency}, we define
\begin{equation}\label{eqn:Mr}
	\begin{split}
		&\underline{r}(k) \triangleq \min_{i\in\R} r_i(k),\; \bar{r}(k) \triangleq \max_{i\in\R} r_i(k),\\
		&m_r(k) \triangleq \min\{\underline{r}(k-\bar{k}+1),\underline{r}(k-\bar{k}+2),\cdots,\underline{r}(k)\}, \\
		&M_r(k) \triangleq \max\{\bar{r}(k-\bar{k}+1),\bar{r}(k-\bar{k}+2),\cdots,\bar{r}(k)\},  
	\end{split}
\end{equation}
where $r_i(k-\ell)=r_i(0)$ if $k<\ell$. Then we denote
\begin{equation}\label{eqn:Pk}
	V(k) \triangleq M_r(k)-m_r(k).
\end{equation}
It is easy to see that
\begin{equation}\label{eqn:P0}
V(0) =\Delta(0).
\end{equation}

Next, let us introduce the following lemma:

	\begin{lemma}\label{lmm:Delta}
		Suppose that Assumptions~\ref{assum:attackmodel} and \ref{ass:network} hold and the initial states of normal oscillators satisfy 
		\begin{equation}\label{eqn:initial}
			\Delta(0)+\frac{4NR}{\alpha^{2NR}}\delta(0)<0.5.
		\end{equation}
		If the normal nodes perform Algorithm~\ref{alg:synchronize} under stealthy attacks, then it holds at any instant $k$ that
		\begin{equation}\label{eqn:Delta}
			\Delta(k)<0.5.
		\end{equation}
	\end{lemma}
	\begin{proof}
	We will prove this lemma by induction. Clearly, by \eqref{eqn:initialphase}, $\Delta(0)<0.5$ holds. Then suppose that at some instant $k\in\mathbb{Z}_{>0}$, it holds for all $t< k$ that $\Delta(t)<0.5$. 
	
An event occurs at the $k$-th instant. That is, there exists a node, labeled as $i$, such that one of the following scenarios is true:
	
	\begin{enumerate}
		\item[\;]  \textit{Scenario I:} Node $i$ is normal and emits a pulse at the $k$-th instant. That is, $\phi_{i}(k)=0$.
		\item[\;]  \textit{Scenario II:} Node $i$ is normal and 
		makes an update at the $k$-th instant. That is, $\phi_{i}(k)=0.5$.
		\item[\;]  \textit{Scenario III:} Node $i$ is misbehaving and emits a false pulse at the $k$-th instant.
	\end{enumerate}
	We shall show $\Delta(k)<0.5$ by respectively discussing the three scenarios. 
	
	\textit{Scenario I:} In this scenario, no normal node updates at time $k$. Therefore, every normal oscillator evolves continuously from time $k-1$ to time $k$. As shown in Fig~\ref{fig:evolve1}, two possible cases might occur then. Particularly, in the first case, we denote by node $q$ the normal node that is mostly ahead of $i$ at time $k$. One thus has
	\begin{equation*}
	\begin{split}
	\phi_q(k) &= \frac{1-\phi_i(k-1)}{\omega_i(k-1)}\omega_q(k-1)+\phi_q(k-1),\\
	\phi_s(k) &= \frac{1-\phi_i(k-1)}{\omega_i(k-1)}\omega_s(k-1)+\phi_s(k-1),
	\end{split}
	\end{equation*}
	and
	\begin{equation}\label{eqn:evolve2}
	\begin{split}
	\bar{r}(k)&= 1-\phi_s(k)+\phi_q(k)\\&=1-\Big(\frac{1-\phi_i(k-1)}{\omega_i(k-1)}\omega_s(k-1)+\phi_s(k-1)\Big)\\&\quad\;\;\;+\Big(\frac{1-\phi_i(k-1)}{\omega_i(k-1)}\omega_q(k-1)+\phi_q(k-1)\Big)\\
	&=1-\phi_s(k-1)+\phi_q(k-1)\\&\quad\;\;\;+(1-\phi_i(k-1))\frac{\omega_q(k-1)-\omega_s(k-1)}{\omega_i(k-1)}\\
	&\leq \bar{r}(k-1)+0.5\delta(k-1)\leq M_r(k-1)+0.5\delta(k-1),
	\end{split}
	\end{equation}
	where the last equality holds since $\Delta(k-1)<0.5$ and thus $0\leq 1-\phi_i(k-1)<0.5$.
	
	Similarly, in the second case, we have
	\begin{equation*}
	\begin{split}
	\phi_i(k) = 0,\; \phi_s(k) = \frac{1-\phi_i(k-1)}{\omega_i(k-1)}\omega_s(k-1)+\phi_s(k-1).
	\end{split}
	\end{equation*} 
	Hence, it follows that
	\begin{equation}\label{eqn:evolve3}
	\begin{split}
\bar{r}(k) &= 1-\phi_s(k) \\
	&=1-(1-\phi_i(k-1))\frac{\omega_s(k-1)-\omega_i(k-1)+\omega_i(k-1)}{\omega_i(k-1)}\\&\quad\;\;\;-\phi_s(k-1)\\
	&= \phi_i(k-1)-\phi_s(k-1)\\&\quad\;\;\;+(1-\phi_i(k-1))\frac{\omega_i(k-1)-\omega_s(k-1)}{\omega_i(k-1)}\\
	&\leq \bar{r}(k-1)+0.5\delta(k-1)\leq M_r(k-1)+0.5\delta(k-1).
	\end{split}
	\end{equation}
	
	Therefore, in either case, we can conclude that
	\begin{equation}\label{eqn:M1}
	M_r(k)=\max\{M_r(k-1),\bar{r}(k)\}\leq  M_r(k-1) +0.5\delta(k-1).
	\end{equation}
	
	Moreover, following similar arguments as in the proof of Lemma~\ref{lmm:alphabeta}, it is easy to see that
	\begin{equation}\label{eqn:scenario1}
	m_r(k)\geq  m_r(k-1).
	\end{equation}
		\begin{figure}[]
		\centering
		\includegraphics[width=0.28\textwidth]{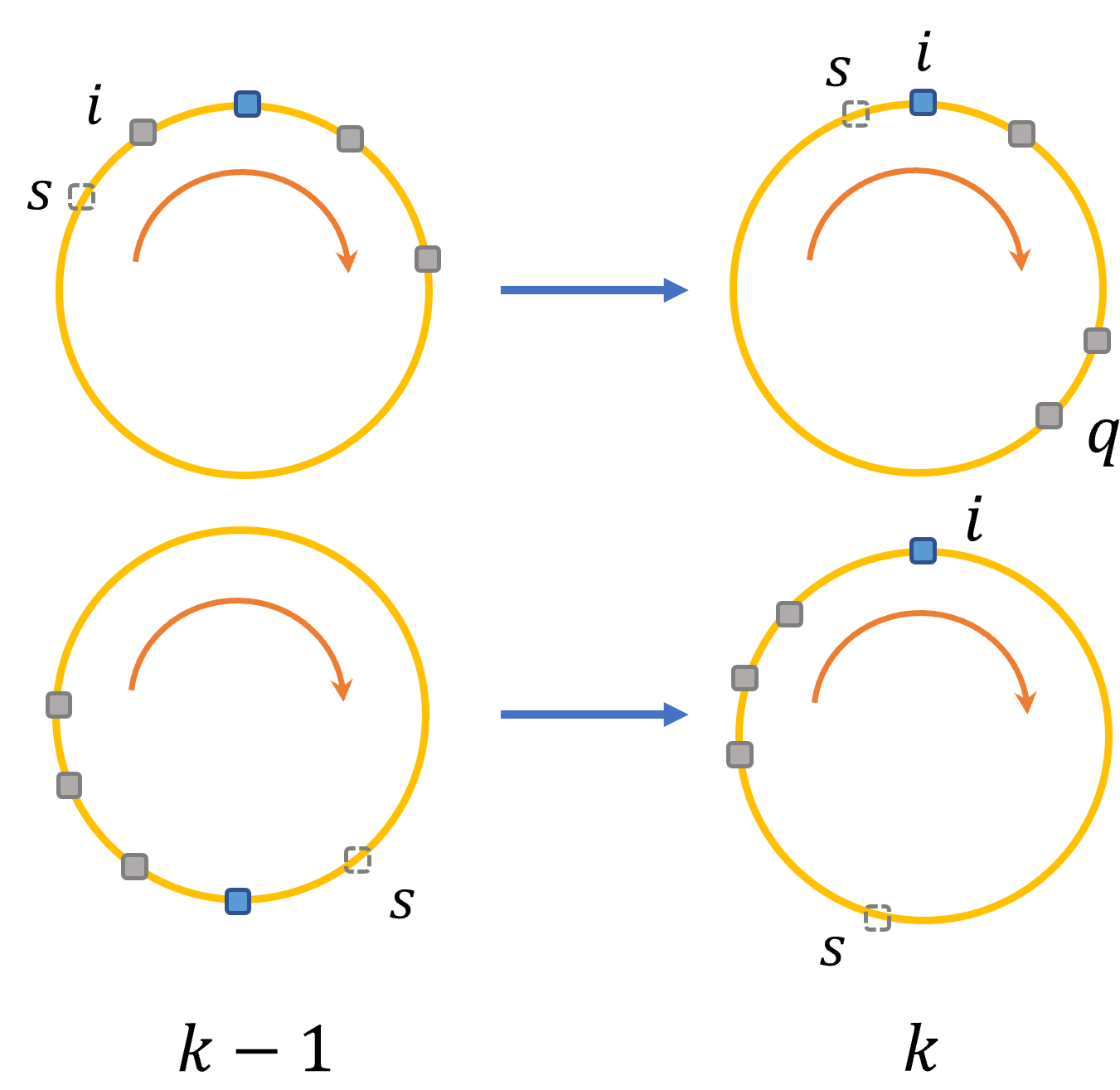}
		\caption{Possible cases if node $i$ emits a pulse at time $k$. The blue square represents the node emitting a pulse or making an update at the corresponding time, the gray squares represent other oscillators, and the dashed square is the virtual node $s$.}
		\label{fig:evolve1}
	\end{figure}
	
	\textit{Scenario II:} In the second scenario, a node $i\in\R$ makes an update at the $k$-th instant, i.e., $\phi_{i}(k)=0.5$.  Since $\Delta(t)<0.5$ holds at all $t< k$, node $i$ receives one and only one pulse from each of its normal in-neighbors within every round. Let us recall the nodes $\alpha$ and $\beta$ defined in the proof of Lemma~\ref{lmm:alphabeta}. We will show 
		\begin{equation}\label{eqn:r}
			r_\alpha (k_1) \in [m_r(k),M_r(k)],\; r_\beta (k_2) \in [m_r(k),M_r(k)],
		\end{equation}
		where $k_1$ and $k_2$ are also defined in the proof of Lemma~\ref{lmm:alphabeta}.
		
		To see this, consider the in-neighbors of node $i$. Under the $f$-total attack model, at time $k_1$, the normal in-neighbor that has the largest relative phase with the virtual node $s$, is the one emitting either the $(f+1)$-th pulse or one of the first $f$ ones. In either case, one can verify that $r_\alpha (k_1) \in [\underline{r}(k_1),\bar{r}(k_1)]$, as node $i$ ignores the first $f$ pulses. Similarly, regardless of whether the adversarial nodes emit pulses or not, among
		all the normal in-neighbors of node $i$, the one taking the least relative phase with $s$ will be the $(d_i-f)$-th one or later to emit a pulse. Therefore, we have $r_\beta (k_2) \in [\underline{r}(k_2),\bar{r}(k_2)]$. On the other hand, notice that node $i$ makes at least one update within $\bar{k}$ steps. Hence, one can verify that 
		\begin{equation}
		 k-\bar{k} <k_1 \leq k_2 < k. 
		\end{equation}
		We have therefore proven \eqref{eqn:r}.
		
		On the other hand, by Algorithm~\ref{alg:synchronize} and Lemma~\ref{lmm:alphabeta}, it can be verified that
		\begin{equation}
			\bar{z}_i(k) = r_\alpha(k_1) - r_i(k_1),\;\underline{z}_i(k) = r_\beta(k_2) - r_i(k_2).
		\end{equation}
		We thus obtain that
		\begin{equation}\label{eqn:rupdate}
			\begin{split}
				r_i(k) &= r_i(k^-) + \frac{1}{2}\bar{z}_i(k)+\frac{1}{2}\underline{z}_i(k)\\
				&=r_i(k^-)+\frac{1}{2}(r_\alpha(k_1) - r_i(k_1))+\frac{1}{2}(r_\beta(k_2) - r_i(k_2)).
			\end{split}
		\end{equation}
	
		For simplicity, let us arrange the instants between $k_1$ and $k$ as
		$k_1=t_0\leq t_1\leq \cdots\leq t_m = k$. By \eqref{eqn:M}, it follows that $1\leq\omega_{s}(t_0)\leq \omega_{s}(t_1)\leq \cdots\leq \omega_{s}(t_m)$, and $\phi_{i}(t_{m})=0.5$.
		Then we can obtain
		 that
		 \begin{equation}\label{eqn:r_i1}
		 \begin{split}
		 	r_i(k^-) = &\;\phi_{i}(t_{m})-\Bigg(\phi_s(t_0)+ \frac{\dist(\phi_{i}(t_{1}),\phi_{i}(t_{0}))}{\omega_{i}(t_0)}\omega_{s}(t_0)\\&+\frac{\dist(\phi_{i}(t_{2}),\phi_{i}(t_{1}))}{\omega_{i}(t_{1})}\omega_{s}(t_{1})+\cdots\\&+\frac{\dist(\phi_{i}(t_{m}),\phi_{i}(t_{m-1}))}{\omega_{i}(t_{m-1})}\omega_{s}(t_{m-1})-1\Bigg).
		 \end{split}
		 \end{equation}
	 Moreover, notice at any $t_\ell$, it follows that 
	 \begin{equation}\label{eqn:r_i2}
	 	\begin{split}
	 	&-\frac{\dist(\phi_{i}(t_{\ell+1}),\phi_{i}(t_{\ell}))}{\omega_{i}(t_\ell)}\omega_{s}(t_\ell)\\=&\dist(\phi_{i}(t_{\ell+1}),\phi_{i}(t_{\ell}))\frac{\omega_i(t_\ell)-\omega_s(t_\ell)-\omega_i(t_\ell)}{\omega_i(t_\ell)}\\\leq& \dist(\phi_{i}(t_{\ell+1}),\phi_{i}(t_{\ell}))(\delta(t_\ell)-1),
	 	\end{split}
	 \end{equation}
	 which holds since $\omega_i(t_\ell)\geq 1.$ Combining \eqref{eqn:r_i1} and \eqref{eqn:r_i2}, we have
	 \begin{equation}\label{eqn:r_i}
	 	r_i(k^-)\leq r_i(k_1)+\delta(k_1)\leq r_i(k_1) +\delta(k-\bar{k}),
	 \end{equation}
 where the first equality holds because node $i$ evolves less than one round from $t_\ell$ to $t_m$ and the last inequality holds by \eqref{eqn:M}. On the other hand, it is easy to see that
		\begin{equation}
			r_i(k^-) \geq r_i(k_1). 
		\end{equation}
Following similar arguments, one can also obtain that 
	\begin{equation}
	r_i(k^-)\leq r_i(k_2) +\delta(k-\bar{k}), 
	\end{equation}
and \begin{equation}
	r_i(k^-) \geq r_i(k_2). 
\end{equation}

		In view of \eqref{eqn:rupdate}, one thus has
		\begin{equation}\label{eqn:rupper}
			\begin{split}
				r_i(k) 
				&\leq \frac{1}{2}(r_i(k_1)+\delta(k-\bar{k}))+\frac{1}{2}(r_i(k_2)+\delta(k-\bar{k}))\\&\qquad+\frac{1}{2}(r_\alpha(k_1) - r_i(k_1))+\frac{1}{2}(r_\beta(k_2) - r_i(k_2))\\&=\frac{1}{2}r_\alpha(k_1)+\frac{1}{2} r_\beta(k_2) +\delta(k-\bar{k})\\&\leq M_r(k-1)+\delta(k-\bar{k}),
			\end{split}
		\end{equation}
		and
		\begin{equation}\label{eqn:rlower}
			\begin{split}
				r_i(k) \geq  \frac{1}{2}r_\alpha(k_1)+\frac{1}{2} r_\beta(k_2) \geq m_r(k-1).
			\end{split}
		\end{equation}
	
	On the other hand, for any other normal agent $j$ which does not update at time $k$, since it evolves freely from time $k-1$ to $k$, we can conclude that 
	\begin{equation}\label{eqn:rj}
	\begin{split}
	r_j(k) &\leq r_j(k-1)+0.5\delta(k-1)\leq M_r(k)+0.5\delta(k-\bar{k}),\\
	r_i(k) &\geq m_r(k).
	\end{split}
	\end{equation}
	Combining \eqref{eqn:rupper} to \eqref{eqn:rj}, it follows that
		\begin{equation}\label{eqn:M2}
			M_r(k)\leq  M_r(k-1) +\delta(k-\bar{k}),\, m_r(k)\geq  m_r(k-1).
		\end{equation}

	\textit{Scenario III:} Finally, we consider the scenario where a false pulse is emitted at time $k$. As that in \textit{Scenario I}, no normal oscillator makes an update at this time, and each of them evolves continuously for less than half a round from time $k-1$ to $k$. We therefore can follow arguments similar to \eqref{eqn:evolve2}--\eqref{eqn:scenario1} and conclude that
		\begin{equation}\label{eqn:M3}
			M_r(k)\leq  M_r(k-1) +0.5\delta(k-1),\, m_r(k)\geq  m_r(k-1).
		\end{equation}
		
		By summarizing the three scenarios, one concludes that \eqref{eqn:M2} always holds and thus
		\begin{equation}
		V(k) \leq V(k-1)+\delta(k-\bar{k}),
		\end{equation}
		where $V(k)$ is defined in \eqref{eqn:Pk}. For simplicity, let us denote 
		\begin{equation}
			\beta \triangleq 1-\frac{\alpha^{\bar{k}R}}{2},
		\end{equation}
		which is convergence rate of the frequency difference given in \eqref{eqn:epsilon}. In view of Lemma~\ref{lmm:frequency}, we thus have
		\begin{equation}\label{eqn:V}
		\begin{split}
		V(k)&\leq V(0)+\sum_{t=0}^{k-\bar{k}}\delta(t)\leq \Delta(0)+\sum_{t=0}^{k-\bar{k}}\beta^{\lfloor t/(\bar{k}R)\rfloor}\delta(0)\\&\leq\Delta(0)+\frac{\bar{k}R}{1-\beta}\delta(0)=\Delta(0)+\frac{2\bar{k}R}{\alpha^{\bar{k}R}}\delta(0),
		\end{split}
		\end{equation}
		where the second inequality holds by \eqref{eqn:epsilon} and \eqref{eqn:P0}. Combining \eqref{eqn:kbar} with \eqref{eqn:initial}, we know
		\begin{equation}
		V(k)<0.5.
		\end{equation}
		Finally, recalling Lemma~\ref{lmm:alphabeta}, we obtain that
		\begin{equation}
		\Delta(k)\leq \bar{r}(k)-\underline{r}(k)\leq V(k)<0.5,
		\end{equation}
		which completes the proof.
	\end{proof}
	
With the above preparations, we are now ready to state the main results of this paper as follows:
	\begin{theorem}\label{thm:frequency}
		Suppose that Assumptions~\ref{assum:attackmodel} and \ref{ass:network} hold and the initial states of normal oscillators satisfy \eqref{eqn:initial}.
	If the normal nodes perform Algorithm~\ref{alg:synchronize} under stealthy attacks, then both the frequencies and phases of normal nodes exponentially reach synchronization. That is,
		\begin{equation}\label{eqn:freqsyn}
			\lim_{t\to\infty} \epsilon(t) = 0,
		\end{equation}
		and
		\begin{equation}\label{eqn:phasyn}
		\lim_{t\to\infty} \Delta(t) = 0.
		\end{equation}
	\end{theorem}
	\begin{proof}
		Combining Lemmas~\ref{lmm:frequency} and \ref{lmm:Delta}, we conclude that \eqref{eqn:epsilon} holds for any $t\in\mathbb{Z}_{+}$. Since $0<1-\frac{\alpha^{\bar{k}R}}{2}<1$, we can arrive at \eqref{eqn:freqsyn}.
	
		On the other hand, in order to show \eqref{eqn:phasyn}, we follow a structure similar to that of Lemma~\ref{lmm:frequency}. However, subtle modifications will be taken accounting for the coupling between the phase and frequency. 
		Given any fixed $t\geq 0$, suppose $V(t)>0$. Let $\varepsilon$ and $\sigma$ be two positive constants such that
		\begin{equation}\label{eqn:sigma}
		\begin{split}
		\varepsilon+\sigma < 0.5, \; \sigma<\frac{0.5^{\bar{k}R}}{1-0.5^{\bar{k}R}}\varepsilon.
		\end{split}
		\end{equation} 
		Moreover, for any $\bar{t}\in\{t,t+1,\cdots, t+\bar{k}R-1\}$, we recursively define
		\begin{equation}\label{eqn:epsilondef}
		\varepsilon(\bar{t}+1)\triangleq 0.5\varepsilon(\bar{t})-0.5\sigma(t),
		\end{equation}
		where 
		\begin{equation}
			\varepsilon(t) \triangleq \varepsilon V(t),\; \sigma(t) \triangleq \sigma V(t).
		\end{equation}
		One can verify that
		\begin{equation}
		\begin{split}
			\varepsilon(t+\bar{k}R) &= 0.5^{\bar{k}R}\varepsilon(t)-\sigma(t)\sum_{k=1}^{\bar{k}R}0.5^k \\&=\Big(0.5^{\bar{k}R}\varepsilon-(1-0.5^{\bar{k}R}) \sigma\Big)V(t)>0,
		\end{split}
		\end{equation}
		where the last inequality holds by \eqref{eqn:sigma}. Combining it with \eqref{eqn:epsilondef} yields
		\begin{equation}
		0<\varepsilon(\bar{t}+1)<\varepsilon(\bar{t}).
		\end{equation}
		Moreover, let us introduce sequences $\{\overline{w}(\bar{t})\}$ and $\{\underline{w}(\bar{t})\}$ by defining
		\begin{equation}
		\begin{split}
		\overline{w}(\bar{t}+1) &\triangleq \overline{w}(\bar{t}) + \delta(\bar{t}-\bar{k}),\;\overline{w}(t) \triangleq M_r(t)-\sigma(t),\\
		\underline{w}(\bar{t}+1) &\triangleq \underline{w}(\bar{t})- \delta(\bar{t}-\bar{k}),\;\underline{w}(t) \triangleq m_r(t)+\sigma(t).\\
		\end{split}
		\end{equation}
		We will show that the following inequality holds:
		\begin{equation}\label{eqn:barw}
			M_r(\bar{t})\leq \overline{w}(\bar{t}) + \sigma(t).
		\end{equation}
		We prove this by induction. Clearly, it is satisfied at $\bar{t}=t$. Suppose it also holds at some $\bar{t}$. By \eqref{eqn:M2}, we obtain that
		\begin{equation}
			M_r(\bar{t}+1)\leq M_r(\bar{t})+ \delta(\bar{t}-\bar{k})\leq \overline{w}(\bar{t}+1) + \sigma(t),
		\end{equation}
		which proves \eqref{eqn:barw}. Similarly, one can verify that
		\begin{equation}
			m_r(\bar{t})\geq \underline{w}(\bar{t}) - \sigma(t).
		\end{equation}
		
		Next, let us define two sets for any $\epsilon\in\mathbb{R}$ and $\bar{t}\geq t$:
		\begin{equation}
		\begin{split}
		\mathcal{R}_r^M(\bar{t},\varepsilon(\bar{t}))\triangleq\{i\in\mathcal{R}: r_{i}(\bar{t})>\overline{w}(\bar{t})-\varepsilon(\bar{t})\},\\
		\mathcal{R}_r^m(\bar{t},\varepsilon(\bar{t}))\triangleq\{i\in\mathcal{R}: r_{i}(\bar{t})<\underline{w}(\bar{t})+\varepsilon(\bar{t})\}.
		\end{split}
		\end{equation}
		We claim that these sets are always disjoint. To see this, notice that for any $i\in \mathcal{R}_r^M(\bar{t},\varepsilon(\bar{t}))$, it holds that
		\begin{equation}
		\begin{split}
		r_{i}(\bar{t})&>\overline{w}(\bar{t})-\varepsilon(\bar{t})> M_r(t)-\sigma(t)-\varepsilon(t)\\&=M_r(t)-(\sigma+\varepsilon)V(t)+\sum_{k=t}^{\bar{t}-1}\delta(k-\bar{k}).
		\end{split}
		\end{equation} 
		Similarly, any $j\in \mathcal{R}_r^m(\bar{t},\varepsilon(\bar{t}))$ has
		\begin{equation}
		r_{j}(\bar{t})<\underline{w}(\bar{t})+\varepsilon(\bar{t})<m_r(t)+(\sigma+\epsilon)V(t)+\sum_{k=t}^{\bar{t}-1}\delta(k-\bar{k}).
		\end{equation} 
		Since $\sigma+\epsilon<0.5$, we conclude that the two sets do not intersect at any $\bar{t}\geq t$.
		
		Then let us consider the sets $\mathcal{R}_r^M(t,\varepsilon(t))$ and $\mathcal{R}_r^m(t,\varepsilon(t))$. Suppose that both of them are nonempty. We shall next prove that,  $\mathcal{R}_r^M(\bar{t},\varepsilon(\bar{t}))$ is nonincreasing in the number of its elements along $\bar{t}=t,t+1,\cdots,t+\bar{k}$. To see this, it is sufficient to show that, any normal node $j$ not in the set $\mathcal{R}_r^M(\bar{t},\varepsilon(\bar{t}))$ also not in the set $\mathcal{R}_r^M(\bar{t}+1,\varepsilon(\bar{t}+1))$. 
		This clearly holds if node $j$ does not make an update at time $\bar{t}$ since by \eqref{eqn:M1} and \eqref{eqn:M3}, it follows that
		\begin{equation}
		\begin{split}
			r_j(\bar{t}+1)&\leq r_j(\bar{t})+0.5\delta(\bar{t})\leq \overline{w}(\bar{t})-\varepsilon(\bar{t})+0.5\delta(\bar{t})\\&<\overline{w}(\bar{t}+1)-\varepsilon(\bar{t}+1).
			\end{split}
		\end{equation}
		On the other hand, let us consider the case where node $j$ makes an update. Let $t_{ij}(\bar{t})$ be the time that it takes for node $i$ to first emits a pulse and then node $j$ makes an update at time $\bar{t}$. By Algorithm~\ref{alg:synchronize}, it is not difficult to conclude that $r_\beta(\bar{t}-t_{\beta j}(\bar{t}))\leq r_j(\bar{t}-t_{\beta j}(\bar{t}))\leq r_j(\bar{t})$. It thus follows from \eqref{eqn:rupper} that
		\begin{equation}\label{eqn:rupperbound}
			\begin{split}
				r_j\left(\bar{t}+1\right)&\leq 0.5 r_\alpha(\bar{t}-t_{\alpha j}(\bar{t}))+0.5 r_\beta(\bar{t}-t_{\beta j}(\bar{t})) +\delta(\bar{t}-\bar{k})\\
				&\leq 0.5 M_r(\bar{t})+ 0.5 [\overline{w}(\bar{t})-\varepsilon(\bar{t})] +\delta(\bar{t}-\bar{k})\\
				&\leq 0.5 [\overline{w}(\bar{t}) + \sigma(t)]+ 0.5 [\overline{w}(\bar{t})-\varepsilon(\bar{t})] +\delta(\bar{t}-\bar{k})\\&= \overline{w}(\bar{t}) +0.5[\sigma(t)-\varepsilon(\bar{t})]+\delta(\bar{t}-\bar{k})\\&=\overline{w}(\bar{t}+1)-\varepsilon(\bar{t}+1).
			\end{split}
		\end{equation}
		Similarly, we can check that the cardinality of $\mathcal{R}_r^m(\bar{t},\varepsilon(\bar{t}))$ is also nonincreasing in $\bar{t}=t,t+1,\cdots,t+\bar{k}$. To see this, let us consider any normal node $j$ not in the set $\mathcal{R}_r^m(\bar{t},\varepsilon(\bar{t}))$. That is, $r_j(\bar{t})\geq \underline{w}(\bar{t})+\varepsilon(\bar{t})$. Suppose it makes an update at $\bar{t}$. In view of \eqref{eqn:r_i}, it follows that
		\begin{equation}
			\begin{split}
				r_\alpha(\bar{t}-t_{\alpha j}(\bar{t}))&\geq r_j(\bar{t}-t_{\alpha j}(\bar{t})) \\&\geq r_j(\bar{t}) -\delta(\bar{t}-t_{\alpha j}(\bar{t})).
			\end{split}
		\end{equation}
		Therefore, one knows that
		\begin{equation}\label{eqn:rlowerbound}
			\begin{split}
				r_j\left(\bar{t}+1\right)&\geq 0.5r_\alpha(\bar{t}-t_{\alpha j}(\bar{t}))+0.5 r_\beta(\bar{t}-t_{\beta j}(\bar{t}))\\
				&\geq 0.5 [\underline{w}(\bar{t})+\varepsilon(\bar{t})-\delta(\bar{t}-t_{\alpha j}(\bar{t}))]+0.5 m_r(\bar{t})\\
				&\geq m_r(\bar{t})+0.5\varepsilon(\bar{t})-0.5\delta(\bar{t}-\bar{k})\\&\geq \underline{w}(\bar{t}+1)+\varepsilon(\bar{t}+1).
			\end{split}
		\end{equation}
		That is, node $j$ is also not in the set $\mathcal{R}_r^m(\bar{t}+1,t,\varepsilon(\bar{t}+1))$.
		
		Next, we show that at least one of the sets $\mathcal{R}_r^M(\bar{t},\varepsilon(\bar{t}))$ and $\mathcal{R}_r^m(\bar{t},\varepsilon(\bar{t}))$ will become empty in finite time. To see this, since the network is $(2f+1)$-robust,  there exists a normal node, labeled as $i$, in either $\mathcal{R}_r^M(t,\varepsilon(t))$ or $\mathcal{R}_r^m(t,\varepsilon(t))$, such that it has at least $2f+1$ in-neighbors from outside the set it belongs to.
		
		Without loss of generality, let $i\in \mathcal{R}_r^M(t,\varepsilon(t))$. It has no less than $2f+1$ neighbors in $\mathcal{V}\backslash\mathcal{R}_r^M(t,\varepsilon(t))$, where at most $f$ of them are from misbehaving nodes. Therefore, at least $f+1$ of these in-neighbors are normal. As discussed above, these nodes will remain outside of $\mathcal{R}_r^M(\bar{t},\varepsilon(\bar{t}))$ for $\bar{t}=t,\cdots,t+\bar{k}$, and one of them would be chosen as node $\beta$.  Suppose that node $i$ makes the update at $\bar{t}=t,\cdots,t+\bar{k}$. We can follow a similar arguments as \eqref{eqn:rupperbound} and obtain that $r_i\left(\bar{t}+1\right)\leq \overline{w}(\bar{t}+1)-\varepsilon(\bar{t}+1)$. This indicates that node $i$ is outside of $\mathcal{R}_r^M(t+\bar{k},\varepsilon(t+\bar{k}))$.
		
		Similarly, if $i\in \mathcal{R}_r^m(t,\epsilon(t))$, we obtain an analogous result that node $i$ will be outside of $\mathcal{R}_r^m(t+\bar{k},\epsilon(t+\bar{k}))$. We therefore conclude that $|\mathcal{R}_r^M(t+\bar{k},\epsilon(t+\bar{k}))|+|\mathcal{R}_r^m(t+\bar{k},\epsilon(t+\bar{k}))|<|\mathcal{R}_r^M(t,\epsilon(t))|+|\mathcal{R}_r^m(t,\epsilon(t))|.$
		
		As long as both $\mathcal{R}_r^M(t+\bar{k},\epsilon(t+\bar{k}))$ and $\mathcal{R}_r^m(t+\bar{k},\epsilon(t+\bar{k}))$ are nonempty, we can repeat the above analysis and conclude that at least one of these two sets will shrink by $1$ after $\bar{k}$ time steps. Moreover, since both $|\mathcal{R}_r^M(t,\epsilon(t))|$ and $|\mathcal{R}_r^m(t,\epsilon(t))|$ are upper bounded by $R$, one of these sets would be empty after $\bar{k}R$ steps. Namely, one of following statements must be true:
		\begin{enumerate}
			\item $\mathcal{R}_r^M(t+\bar{k}R,\epsilon(t+\bar{k}R))=\varnothing,$
			\item $\mathcal{R}_r^m(t+\bar{k}R,\epsilon(t+\bar{k}R))=\varnothing.$
		\end{enumerate}
		Without loss of generality, we assume that the first statement holds. In this case, all normal nodes are outside the set $\mathcal{R}_r^M(t+\bar{k}R,\epsilon(t+\bar{k}R))$. We therefore conclude that
		\begin{equation}
			M_r(t+\bar{k}(R+1))\leq \overline{w}(t+\bar{k}(R+1))-\varepsilon(t+\bar{k}(R+1)).
		\end{equation}
		On the other hand, one knows
		$m_r(t+\bar{k}(R+1))\geq m_r(t).$
		Therefore, one concludes that                                                                                           
		\begin{equation}
			\begin{split}
				&V(t+\bar{k}(R+1))\\&\leq \overline{w}(t+\bar{k}(R+1))-\varepsilon(t+\bar{k}(R+1))-m_r(t)\\&\leq (1-\sigma)V(t)+\gamma\delta(t-\bar{k}).
			\end{split}
		\end{equation}
		where $\gamma =\bar{k}(R+1)$. Hence, it follows that
		\begin{equation}\label{eqn:error}
			\begin{split}
				V&(t+r\bar{k}(R+1))\leq (1-\sigma)^rV(t)\\&+\gamma\sum_{l=0}^{r-1}(1-\sigma)^{r-1-l}\delta(t+l\bar{k}R-\bar{k}).
			\end{split}
		\end{equation}
		As proved in Theorem~\ref{thm:frequency}, we conclude that $$\lim_{t\rightarrow \infty}\delta(t)=0.$$ In view of \cite[Lemma 7]{nedic2010constrained}, the second term in the right hand side of \eqref{eqn:error} goes to $0$ as $r\to\infty$. Therefore, it holds that
		\begin{equation}
			\lim_{r\rightarrow \infty}V(t+r\bar{k}(R+1))=0.
		\end{equation}
		Since the above equation holds true for any $t\in\mathbb{Z}_{\geq 0}$, one concludes that
		\begin{equation}
			\lim\limits_{k\to\infty}V(k) = 0.
		\end{equation}
		Combining this with Lemma~\ref{lmm:alphabeta}, we finally complete the proof.
	\end{proof}

Based on the results presented above, we can conclude that Algorithm~\ref{alg:synchronize} is resilient to network misbehaviors and node failures:
	\begin{theorem}\label{thm:main}
		Suppose that Assumptions~\ref{assum:attackmodel} and \ref{ass:network} hold and the initial states of normal oscillators satisfy \eqref{eqn:initial}.
With Algorithm~\ref{alg:synchronize}, all the normal nodes either achieve the resilient synchronization or detect the presence of the misbehaving nodes. 
	\end{theorem}
\begin{proof}
As previously discussed, the initial condition \eqref{eqn:initial} ensures that each normal oscillator fires only once within each round. This means that if any misbehaving node fires more than once within any single round, its normal neighbors running Algorithm~\ref{alg:synchronize} can directly detect the anomaly.

On the other hand, if all misbehaving nodes perform stealthily, Theorem~\ref{thm:frequency} guarantees that synchronization of both frequencies and pulses is achieved among normal nodes, even though they are unaware of the presence of misbehaviors.

Finally, the safety condition defined in Definition~\ref{def:resilientCon} is ensured by \eqref{eqn:epsilon} and \eqref{eqn:Delta}. Therefore, we can conclude that Algorithm~\ref{alg:synchronize} is resilient to misbehaviors and satisfies the safety condition defined in Definition~\ref{def:resilientCon}.
\end{proof}

It is worth noting that when considering homogeneous oscillators with identical frequencies, the initial condition \eqref{eqn:initial} is relaxed to $\Delta(0) < 0.5$, namely, the initial length of the containing arc is less than $0.5$. This observation is consistent with the findings in \cite{iori2021resilient,wang2018pulse,wang2020attack}.
In this context, Algorithm~\ref{alg:synchronize} can be regarded as an extension of existing solutions to a more general case where each oscillator may start with a different frequency.

\begin{remark}
Let us recall the initial condition \eqref{eqn:initial}. From \eqref{eqn:V}, it is not difficult to see that \eqref{eqn:initial} can be relaxed to
\begin{equation}\label{eqn:initial_k}
\Delta(0)+\frac{2\bar{k}R}{\alpha^{\bar{k}R}}\delta(0)<0.5,
\end{equation}
where $\bar{k}\leq 2N$ is a constant such that each normal oscillator makes at least one update during $\bar{k}$ steps.
\end{remark}

\begin{remark}
The initial condition \eqref{eqn:initial} is conservative since we consider the worst-case malicious behaviors in analysis. However, in practical scenarios, the proposed algorithm can function effectively under much more relaxed conditions than \eqref{eqn:initial}. We will provide numerical examples in Section~\ref{sec:simulation} to further demonstrate this point.
\end{remark}

\section{Resilient Synchronization using Relative Frequencies}\label{sec:extend}
Clearly, in Algorithm~\ref{alg:synchronize}, the utilization of packet-based communication is inevitable for transmitting the real-valued absolute frequencies $\{\omega_i\}$. However, since each oscillator operates with its own clock, it is often the case that the absolute frequencies are not directly accessible to individual agents. In such situations, normal oscillators can update their frequencies only using relative values. Inspired by it, this section will propose a resilient synchronization algorithm based on relative frequencies. 

In this approach, every normal oscillator is required to emit two pulses within a single round, specifically when it reaches $1-\zeta$ and $1$, where $0<\zeta<0.5$. We describe the complete algorithm in Algorithm~\ref{alg:relative}. It is worth noting that in this case, oscillator $i$ does not need knowledge of the absolute frequencies. Instead, it adjusts its frequency by a factor of $\eta_i(t)$ relative to its previous frequency, as depicted in Step 4.b) of Algorithm~\ref{alg:relative}. Furthermore, the algorithm can be implemented using simple pulses, eliminating the need for transmitting real-valued messages.

%

	\begin{algorithm}[h!] 
	1:\: The phase $\phi_{i}(t)$ evolves from $0$ to $1$.
	
	2:\: Once $\phi_{i}(t)$ reaches $1$, oscillator $i$ sets $\gamma_i=1$, emits an ``end pulse'', and resets its phase to $0$. On the other hand, if $\phi_{i}(t)=1-\zeta$, oscillator $i$ emits a `` start pulse''.
	
3:\: When oscillator $i$ receives an end pulse, it updates $c_i=c_i+1$. The calculations for $\bar{z}_{i}(t)$ and $\underline{z}_{i}(t)$ are performed in the same manner as described in Step 3 of Algorithm~\ref{alg:synchronize}.

	4:\: When $\phi_{i}(t)=0.5$ and $\gamma_i=1$, oscillator $i$ executes the following:
	
	\begin{enumerate}[leftmargin = 30 pt]
		\item[a):] If $c_{i}>d_{i}$, oscillator $i$ detects an attack and informs the system. It does not make any update and sets $\phi_{i}\left(t\right)=\phi_{i}(t^{-})$ and $\omega_{i}\left(t\right)=\omega_{i}(t^{-})$.
		
		\item[b):] If $c_{i} \leq d_{i}$, oscillator $i$ updates its phase by \eqref{eqn:phaseupdate}. Moreover, it calculates the relative frequency with respect to any of its in-neighbors $j\in\mathcal{N}^+_i$ as
		$$\eta_{ij}(t^-) = \zeta/(\phi_{i}(\hat{t}_j)-\phi_i(\tilde{t}_j)).$$ 
		Here, $\phi_{i}(\hat{t}_j)$ and $\phi_i(\tilde{t}_j)$ represent the phases of oscillator $i$ when it receives the two pulses from oscillator $j$ in the current round.

		Then, it removes $f_i$ largest and $f_i$ smallest values in $\{\eta_{ij}(t^-)\}$, where $f_i$ is given in \eqref{eqn:f0}. Denote by $\mathcal{J}_i(t)$ the set of remaining states. Oscillator $i$ updates its frequency as
		$$	\omega_i(t) = \eta_i(t)\omega_i(t^{-}),$$ where
		\begin{equation}\label{eqn:frequency2}
		\eta_i(t) = a_{ii}(t)+\sum_{j\in \mathcal{J}_i(t)} a_{ij}(t) \eta_{ij}(t^-).
		\end{equation}
		The weights are constrained to have a lower bound of $\alpha>0$ and satisfy the condition $a_{ii}(t)+\sum_{j\in \mathcal{J}_i(t)} a_{ij}(t)=1.$
	\end{enumerate}
	
	5:\: Oscillator $i$ resets $\gamma_i=0$ and  $c_i=0$.\\
	\caption{Resilient synchronization of heterogeneous pulse-coupled oscillators by relative frequencies}
	\label{alg:relative}
\end{algorithm}

We have the following theorem:
\begin{theorem}\label{thm:relative}
		Suppose that Assumptions~\ref{assum:attackmodel} and \ref{ass:network} hold and the initial states of normal oscillators satisfy \begin{equation}\label{eqn:initial2}
		\Delta(0)+\frac{2\bar{k}R}{\alpha^{\bar{k}R}}\delta(0)<0.5-\zeta.
		\end{equation} 
With Algorithm~\ref{alg:relative}, all the normal nodes either achieve the resilient synchronization or detect the presence of the misbehaving nodes. 
\end{theorem}
\begin{proof}
Based on similar arguments as in Lemma~\ref{lmm:Delta}, it is trivial to verify that under \eqref{eqn:initial2}, the following statement holds:
\begin{equation}
\Delta(k)<0.5-\zeta, \;\forall k.
\end{equation} 
This means that the phase difference between any two normal oscillators is always less than $0.5-\zeta$. Consequently, for any normal oscillator $j$, when it transitions from $1-\zeta$ to $1$, there is no phase jump for any of its normal neighbors $i\in\mathcal{N}^-_j$. As such, one obtains that
$$\eta_{ij}(t^-) = \frac{\omega_j(t^-)}{\omega_i(t^-)}.$$
Multiplying $\omega_i(t^-)$ on both sides of \eqref{eqn:frequency2} yields 
$$\omega_i(t)=a_{ii}(t)\omega_i(t^{-})+\sum_{\omega_j(t^-)\in \mathcal{J}_i(t)} a_{ij}(t) \omega_j(t^-),$$
which is identical to \eqref{eqn:frequencyupdate}. Therefore, following similar arguments presented in this paper, we conclude the proof of this theorem.
\end{proof}

Compared to \eqref{eqn:initial_k}, \eqref{eqn:initial2} imposes stricter constraints on the initial states. This is because Algorithm~\ref{alg:synchronize} assumes that the absolute frequencies are accessible through packet-based communications, while Algorithm~\ref{alg:relative} utilizes pulses to calculate relative frequencies, resulting in lower energy consumption. This observation highlights a trade-off between relaxed initial conditions and reduced communication burden.

\section{Numerical Examples}\label{sec:simulation}
Let us consider the network of $8$ oscillators which cooperate over the graph given in Fig.~\ref{fig:network}. As verified in \cite{leblanc2012resilient}, this graph is $3$-robust. Therefore, the oscillators can survive $1$-local attack.  
\begin{figure}[!htbp]
	\centering
	\includegraphics[width=0.18\textwidth]{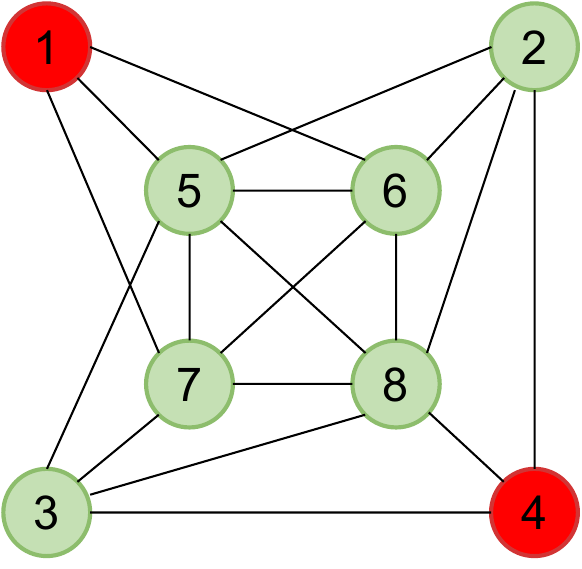}
	\caption{A network that is $3$-robust (\hspace{1pt}\cite{leblanc2012resilient}).}
	\label{fig:network}
\end{figure}

To verify this, let oscillators $1$ and $4$ be misbehaving. They intend to prevent normal oscillators from achieving synchronization by maliciously broadcasting their frequencies as 
\begin{equation}
\begin{split}
\omega_{1}(t)=1+|\sin(t)|,\; \omega_{4}(t)= 1+t-\lfloor t\rfloor.
\end{split}
\end{equation} 
Moreover, they emit attacking pulses as shown in Fig.~\ref{fig:PCO}. 

On the other hand, the normal oscillators are randomly initialized with frequencies from $[1,2]$ and phases from $[0,0.5]$. This initial condition is less strict compared to \eqref{eqn:initial}. However, as shown in Fig.~\ref{fig:PCO}, the normal oscillators can still achieve the resilient synchronization by following Algorithm~\ref{alg:synchronize}.

To further illustrate this point, we conduct a series of tests with different initial phases and frequencies. Specifically, given the range of initial phases $\Delta(0)$, we examine the maximum difference in initial frequencies $\delta(0)$, under which the proposed algorithms can either detect malicious agents or facilitate resilient synchronization among normal agents. For each $\Delta(0)$, we randomly set the initial phases of normal oscillators within $\Delta(0)$, perform the Monte Carlo trials for $1000$ runs, and calculate their average value. The results are presented in Fig.~\ref{fig:initial}. It is clear that as $\Delta(0)$ increases, the corresponding $\delta(0)$ decreases. This indicates that when the normal oscillators are widely spread in terms of initial phases, they allow less frequency difference to achieve synchronization. The findings from Fig.~\ref{fig:initial} further confirm that our algorithms can function effectively under significantly more relaxed conditions than \eqref{eqn:initial}.

\begin{figure}[!htbp]
	\centering
	\input{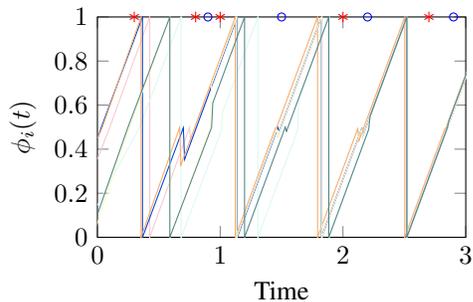}
	\caption{The time responses of the phases of normal oscillators, where the red asterisks and blue circles indicate the timings when the attacking
		pulses are sent from malicious oscillators $1$ and $4$, respectively.}
	\label{fig:PCO}
\end{figure}

\begin{figure}[!htbp]
	\centering
%
%
\begin{tikzpicture}

\begin{axis}[%
width=1.928in,
height=1.154in,
at={(1.011in,0.642in)},
scale only axis,
xmin=0,
xmax=0.5,
xlabel={{$\Delta(0)$}},
xtick = {0, 0.1,0.2,0.3,0.4,0.5},
xticklabels = {$0$,$0.1$,$0.2$,$0.3$,$0.4$,$0.5$},
ymin=1.28,
ymax=2.1,
ylabel={{$\delta(0)$}},
axis background/.style={fill=white}
]

\addplot[only marks, mark=x, mark options={}, mark size=2.0000pt, draw=blue]
table[row sep=crcr]{%
	0	2\\
	0.1	1.8\\
	0.2	1.7\\
	0.3	1.6\\
	0.4	1.35\\
	0.5 1.3\\
};

\end{axis}

\end{tikzpicture}%
	\caption{Maximum difference in initial frequencies for varying initial phase differences among normal oscillators.}
	\label{fig:initial}
\end{figure}
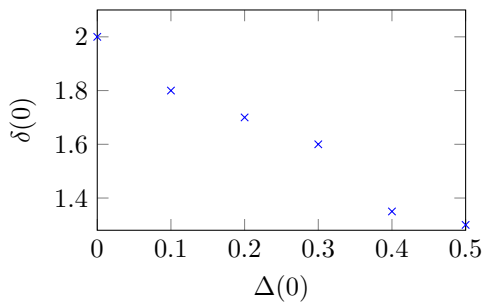

\section{Conclusion}\label{sec:conclude}
This paper has addressed the problem of pulse-coupled synchronization in an adversarial environment, where oscillators in the network have different frequencies and some of them might be manipulated by adversaries. To counter these misbehaviors, we have proposed two resilient synchronization protocols by adapting the MSR algorithm to pulse-based interactions: one relies on packet-based communication to transmit absolute frequencies, while the other operates purely with pulses to calculate relative frequencies. 
Under certain conditions on the initial values and the network topology, our protocols guarantee that normal oscillators can either detect malicious behaviors or achieve exponential synchronization in both phases and frequencies. The trade-off between relaxed initial conditions and reduced communication burden has also been discussed. 

	\bibliographystyle{IEEEtran}
	\bibliography{reference} 
	
\end{document}